\documentclass[11pt]{article}

\usepackage{epsfig} 
\usepackage{amssymb}  
\usepackage{amsmath}

\usepackage{enumitem}
\usepackage{pifont}
\usepackage[dvipsnames]{xcolor}
\usepackage{pgfplots}
\usepackage{algorithm2e}
\usepackage{varwidth}
\usepackage{soul}
\usepackage{verbatim}
\usetikzlibrary{arrows,shapes}

\usepackage{graphics} 
\usepackage{amsmath} 
\usepackage{amssymb}  

\usepackage{amsthm}
\usepackage[scientific-notation=true]{siunitx}

\usepackage{tikz}
\usetikzlibrary{arrows.meta}
\usepackage{dsfont}

\usepackage{amsmath} 
\usepackage{graphicx}
\usepackage{caption}
\usepackage{subcaption}
\usepackage{multirow}
\usepackage[left=3cm,right=3cm]{geometry}

\usepackage[colorlinks,citecolor=blue]{hyperref}

\newtheorem{theorem}{Theorem}

\newtheorem{proposition}{Proposition}
\theoremstyle{definition}
\newtheorem{example}{Example}
\newtheorem{definition}{Definition}

\theoremstyle{remark}
\newtheorem{remark}{Remark}

\usepackage{placeins}
\newcommand{\bb}[1]{\mathbf{#1}}

\tikzset{>=latex}
\definecolor{darkblue}{RGB}{0,101,204}
\definecolor{myorange}{RGB}{255,131,0}

\begin{document}
%
\title{A Game Theoretic Macroscopic Model of Bypassing at Traffic Diverges with Applications to Mixed Autonomy Networks}
%
%
%

\author{Negar Mehr, Ruolin Li, and~Roberto~Horowitz.\footnote{N. Mehr, R. Li and R. Horowitz are with the Department
		of Mechanical Engineering, University of California, Berkeley, Berkeley,
		CA, 94720 USA e-mails: negar.mehr@berkeley.edu, ruolin\_li@berkeley.edu, horowitz@berkeley.edu }}%

\maketitle

\begin{abstract}
Vehicle bypassing is known to negatively affect delays at traffic diverges. However, due to the complexities of this phenomenon, accurate and yet simple models of such lane change maneuvers are hard to develop. In this work, we present a macroscopic model for predicting the number of vehicles that bypass at a traffic diverge. We take into account the selfishness of vehicles in selecting their lanes; every vehicle selects lanes such that its own cost is minimized. We discuss how we model the costs experienced by the vehicles. Then, taking into account the selfish behavior of the vehicles, we model the lane choice of vehicles at a traffic diverge as a Wardrop equilibrium. We state and prove the properties of Wardrop equilibrium in our model. We show that there always exists an equilibrium for our model. Moreover, unlike most nonlinear asymmetrical routing games, we prove that the equilibrium is unique under mild assumptions. We discuss how our model can be easily calibrated by running a simple optimization problem. Using our calibrated model, we validate it through simulation studies and demonstrate that our model successfully predicts the aggregate lane change maneuvers that are performed by vehicles for bypassing at a traffic diverge. We further discuss how our model can be employed to obtain the optimal lane choice behavior of the vehicles, where the social or total cost of vehicles is minimized. Finally, we demonstrate how our model can be utilized in scenarios where a central authority can dictate the lane choice and trajectory of certain vehicles so as to increase the overall vehicle mobility at a traffic diverge. Examples of such scenarios include the case when both human driven and autonomous vehicles coexist in the network. We show how certain decisions of the central authority can affect the total delays in such scenarios via an example.  
\end{abstract}


%

\section{Introduction}\label{intro}
Huge delays and costs are incurred by travelers due to traffic congestion. Thus, it is of paramount importance to derive accurate models of traffic behavior, as such models can be used to analyze traffic networks to gain an insight on how the traffic conditions can be improved in urban and freeway networks. However, since traffic networks normally exhibit very complex behaviors, developing models that are both accurate and simple enough for analysis and traffic management purposes is nontrivial. Among various traffic phenomena, vehicle lane changes are known to significantly affect traffic congestion and delays~\cite{ahn2007freeway}. Thus, it is important to accurately model vehicle lane changes; however, modeling the lane change behavior of vehicles remains among the most difficult traffic flow phenomena, partly due to the fact that they are dependent on the drivers' decision making process.
Furthermore. it is difficult to characterize the negative effects of vehicle lane changes and bypassing on upstream traffic streams. The existing literature on modeling lane change behaviors is mostly divided into two categories: 1) modeling the microscopic lane change decision process of the vehicles, or 2) investigating and quantifying the macroscopic effect of aggregate lane change maneuvers on traffic conditions. 

In regards to the existing research in the first category, lane change behavior of vehicles was first systematically studied in~\cite{gipps1986model}, where a set of rules and conditions were developed under which a single vehicle was assumed to change its lane. The set of derived conditions were assumed to depend on vehicle microscopic parameters such as vehicle velocity and road segment parameters such as how much space is available in the neighboring lanes. Several other microscopic models were derived in~\cite{rorbech1976multilane,yeo2008oversaturated,lee2016probability}. A survey and review of such microscopic models is available in~\cite{zheng2014recent}. Aligned with these models, several car following models were proposed by researchers to model the vehicle microscopic behavior such as acceleration and lane change maneuvers. Examples of such works can be found in~\cite{rothery1992car, aycin1999comparison,kesting2007general}. 
Recently, a game theoretic approach was used in~\cite{talebpour2015modeling} and~\cite{meng2016dynamic} to model the lane change behavior of a single vehicle, where the lane change decision was assumed to be taken by a vehicle for increasing its speed. In~\cite{kita1999merging}, a similar approach was utilized to mimic the behavior of drivers at traffic merges.

With the recent advances in autonomous vehicles technology, a large body of literature has been devoted on how to design and control autonomous vehicles based on these vehicle lane change microscopic models, such that the vehicles exhibit lane change behaviors that are safe, are also similar to the lane change decisions made by humans, and are optimal. In~\cite{dong2017intention}, the intention of drivers when merging to freeway lanes was estimated. In~\cite{ulbrich2013probabilistic}, a decision
making approach for performing lane changes in a
fully automated vehicle driving scenario was presented and evaluated. In~\cite{khan2014analyzing}, the requirements associated with an optimal lane change behavior were described where minimizing fuel consumption and travel time were considered as objectives. In~\cite{mccall2007lane}, computer vision techniques were utilized to infer lane change intents. 

In regard to the existing research on the macroscopic effect of lane changes, there has been a focus on how to quantify the negative effects of lane change maneuvers on upstream traffic congestion. In~\cite{laval2006lane}, lane changing vehicles were modeled as particles endowed with mechanical properties. 
The implications and applications of this model were discussed in~\cite{laval2007impacts}. In~\cite{coifman2005lane}, it was demonstrated via case studies that lane change maneuvers could lead to reductions in freeway capacity. In~\cite{jin2010kinematic}, the impacts of lane change behaviors were modeled via the introduction of lane changing intensity variables and modified macroscopic traffic flow fundamental diagrams. In~\cite{wu2017multi}, a stochastic lane change model was developed for capturing the system level lane changing characteristics. 

In this paper, we develop a novel model of the \emph{aggregate} lane change maneuvers taken by vehicles that perform \emph{bypasses} at traffic diverges. By performing a bypass, it is meant that a vehicle performs a lane change behavior to the lane that corresponds to its intendant route very close to the diverge often cutting in front of vehicles that have made the lane change maneuver far upstream of the diverge. 
We develop a model capable of predicting the fraction of vehicles that perform bypasses at a traffic diverge to take an exit link. We study vehicle bypassing at the \emph{macro} scale, where we predict the number of vehicles who will change their lanes in order to take an appropriate exit. In particular, given the number of vehicles who wish to take a certain exit, our model can predict how many vehicles will perfrom bypassing close to the diverge in order to take an exit. We assume that vehicles act selfishly, i.e. every vehicle decides on its route and lane choice such that its own cost of travel is minimized. We describe how such costs incurred by vehicles can be modeled. Since our focus is on developing a macroscopic fluid like model for the behavior of vehicles, we model the equilibrium that results from the selfishness of vehicles as a Wardropian equilibrium. We prove that our model always has an equilibrium, and further, we use a novel machinery to show that its equilibrium is unique under mild assumptions. We describe how our model can easily be calibrated by solving a mixed integer linear program, and show through simulation studies that our model yields promising results, as it can successfully predict the aggregate bypassing behavior of vehicles at a fork.

Our framework, albeit simple, provides a powerful tool for quantifying the optimal lane change behavior of vehicles at traffic diverges. Our model not only predicts the aggregate bypassing behavior of vehicles but can also be used to quantitatively analyze what the optimal lane choice of vehicles must be. Our model is particularly beneficial in scenarios when a central authority can route a fraction of vehicles such that the overall delay is minimized. For instance, in networks with mixed vehicle autonomy, autonomous vehicles might be routed by a central planner. We demonstrate how our model can be used in such scenarios for deciding on the bypassing behavior of the autonomous vehicles such that the resulting induced equilibrium has a minimally detrimental social cost. To the best of our knowledge, there is no such work in the literature. 

The organization of this paper is as follows. In Section~\ref{sec:model}, we describe our modeling framework. In Section~\ref{sec:eq_prop}, we state and prove the properties of our model. Simulation studies including model calibration and validation are described in Section~\ref{sec:model_val}. In Section~\ref{sec:soc_opt}, we describe how our model can be used for finding the optimal lane choice of vehicles at a traffic diverge. In section~\ref{sec:mixed_aut}, we describe the applications of our model to traffic networks with mixed autonomy. We conclude the paper and discuss future directions in Section~\ref{sec:future}.

\section{The Model} \label{sec:model}

\begin{figure}
\centering
\begin{tikzpicture}
  \fill[gray!30] (0,1) [rounded corners = 10pt]-- (4,1) [rounded corners = 0pt]-- (7,2) -- (7,2/3) -- (5,0) -- (7,-2/3) -- (7,-2) [rounded corners = 3pt]-- (4,-1) [rounded corners = 0]-- (0,-1) -- (0,1);

  \draw[very thick,rounded corners=10pt] (0,1) -- (4,1) -- (7,2);
  \draw[very thick,rounded corners=10pt] (0,-1) -- (4,-1) -- (7,-2);
  \draw[very thick] (7,2/3) -- (5,0) -- (7,-2/3);

  \draw[very thick, dash pattern=on 5pt off 5pt, white] (0,0) -- (4.5,0);

  \node[scale=0.5,rotate=180] at (2,0.5) {\includegraphics{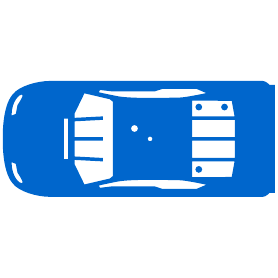}};
  \node[scale=0.5,rotate=200] at (3.5,-0.1) {\includegraphics{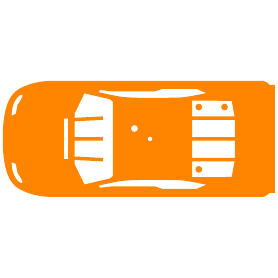}};

  \draw[thick, dashed, darkblue] (0,0.5)  .. controls (4,0.5) and (5,0.6) .. (7,1+1/3);
  \draw[thick, dashed, myorange] (0,-0.5)  .. controls (4,-0.5) and (6,1) .. (7,1+1/3);

  \node[darkblue,scale=0.9] at (0.7,0.75) {$x^s_1$};
  \node[myorange,scale=0.9] at (2,-0.65) {$x^b_1$};
  
  \node at (0.2,0.7) {I};
  \node at (0.2,-0.7) {II};
  \node at (6,1.3) {1};
  \node at (6,-1) {2};
  
\end{tikzpicture}
\caption{Example of a traffic diverge with two destination links 1 and 2. For exit link 1, a steadfast vehicle (blue car) constituting $x_1^s$ and a bypassing vehicle (orange car) forming $x_1^b$ are shown.}
\label{fig:diverge}
\end{figure}
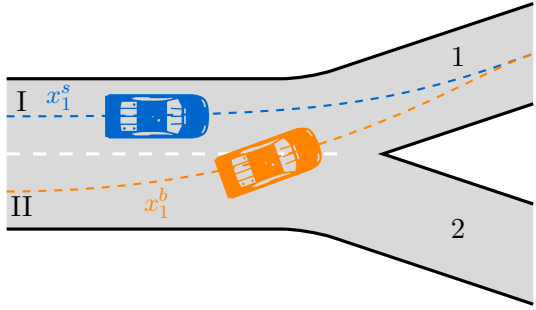

We consider a traffic diverge where a link bifurcates into two links, which is a common scenario for freeway and arterial forks. We wish to study the bypassing behavior of vehicles in such diverges, where certain lanes correspond to a certain route or an exit link. Normally, in these scenarios, among the vehicles with the same target exit link, a fraction of vehicles choose to have the lanes that correspond to their exit, far upstream of the diverge, while the remaining fraction of vehicles perform bypasses and change to the lanes that correspond to their exit links at the distances that are very close to the diverge. More specifically, given the demand of vehicles for each possible exit link, our goal is to derive a macroscopic traffic model that can predict the fraction of vehicles which exhibit either of these two behaviors. Note that we wish to capture the aggregate bypassing behavior of vehicles in a macro scale rather than deriving the conditions under which a single vehicle decides to perform a lane change.

Let $I = \{1,2\}$ be the index set of the exit links at a fork diverge with two exit links (see Figure~\ref{fig:diverge}). Let $d_1$ and $d_2$ be the demands of vehicles that wish to take exit link $1$ and $2$ respectively. Additionally, we use $d = d_1 + d_2$  to represent the total demand of vehicles upstream of the diverge. We use $f_1 = \frac{d_1}{d}$ and $f_2 = \frac{d_2}{d}$ to represent the fraction of vehicles whose destinations are links $1$ and $2$. We describe our model in terms of these normalized flows rather than the actual flows since it simplifies our analysis. Let $F = \{ f_1, f_2\} $ be the normalized demand configuration which is the set of normalized demands for the diverge destinations. 
For each exit link $i \in I$, let $x^s_i$ denote the fraction of ``steadfast'' vehicles among $f_i$, which are the vehicles that take the lanes that correspond to their destination link $i$ far upstream of the diverge and remain on their lanes, whilst $x^b_i$ denotes the fraction of ``bypassing'' vehicles that choose to change their lanes to the lanes that correspond to their exit link $i$ at or at the vicinity of the diverge. Figure~\ref{fig:diverge} illustrates steadfast and bypassing vehicles that wish to take exit link 1. We assume that vehicles change their lanes only once, i.e. if a vehicle is in its target lane, it remains there. We let $\mathbf{x} = (x_i^s, x_i^b: i \in I)$ be the vector of steadfast and bypassing normalized flows for the two possible destinations of a fork. A normalized flow vector $\mathbf{x}$ is feasible for a given normalized demand configuration $F$ if it satisfies
\begin{gather}\label{eq:flow_cons}
f_i = x_i^s + x_i^b, \quad \forall i \in I, \\
x^s_i \geq 0, \; x^b_i \geq 0, \quad \forall i \in I.
\end{gather}
\begin{example}
Consider the diverge shown in Figure~\ref{fig:diverge}.
In this example, there are two upstream freeway lanes I and II which connect to exit links 1 and 2. For this diverge, $x_1^s$ is the fraction of vehicles that remain on lane I and take exit link 1, whereas $x_1^b$ is the fraction of vehicles that move along lane II and change to lane I at vicinity of the diverge.
\end{example}

For each exit link $i \in I$, we assume that all steadfast vehicles constituting $x_i^s$ experience the same travel cost. Likewise, all bypassing vehicles taking an exit link $i$, experience the same travel cost. For each destination link $i \in I$, we let $J_i^s$ and $J_i^b$ be the cost incurred on the vehicles forming $x_i^s$ and $x_i^b$ respectively. It is important to note that for each $i \neq j \in I$, $J_i^s$ or $J_i^b$ depends not only on 
$x_i^s$ and $x_b^a$ but also can depend on $x_j^s$ and $x_j^b$. 
For each $i \neq j \in I$, we model the cost per unit of flow of the steadfast vehicles by
\begin{align}\label{eq:J_i^s}
J_i^s(\mathbf{x}) &= C_i^t \left( x_i^s + x_j^b \right) + C_i^c x_i^b \left ( x_i^s + x_j^b \right),
\end{align}
where $C_i^t$ and $C_i^c$ 
are positive constants. The constant $C_i^t$ is the cost of traversing the lanes that connect to the exit $i$. Since $(x_i^s + x_j^b)$ is the total fraction of vehicles that traverse the lanes that connect to exit $i$, $(x_i^s + x_j^b)$ is multiplied by $C_i^t$ (e.g. a total of $x_1^s+x_2^b$ traverse link $I$ in Figure~\ref{fig:diverge}). It indicates that the more occupied the lanes that correspond to an exit are, the more expensive their traversal is due to the induced congestion.
On the other hand, the constant $C_i^c$ is used to reflect the negative cross effects caused by the lane change behavior of bypassing vehicles $x_i^b$. This term is used to mimic the fact that as the vehicles in $x_i^b$ bypass and change their lanes to take the exit $i$, they use the roads (resources) that join the exit $i$; thus, they will create delays for the vehicles that are already in those lanes. Note that since $x_i^s$ and $x_j^b$ both share the target link of $x^b_i$ up to the vicinity of the diverge, the total fraction of the vehicles present in the target lanes of $x_i^b$ is $( x_i^s + x_j^b )$. Hence, $C_i^c$ is multiplied by $( x_i^s + x_j^b )$ and $x_i^b$. This multiplication implies that the higher the number of vehicles that bypass $x_i^b$ is, or, the more occupied the lanes that join exit $i$ are, the larger the incurred cost will be.


Now, we describe how we model the costs incurred on the bypassing vehicles. For each $i\neq j \in I$, we model $J_i^b$ via


\begin{align}\label{eq:J_i^v}
J_i^b(\mathbf{x}) &= C_j^t \left(x_j^s + \gamma_i x_i^b \right) + C_j^c x_j^b \left(  x_j^s + x_i^b \right) 
\end{align}
where $\gamma_i$ is a constant assumed to satisfy $\gamma_i \geq 1$, and $C_j^t$ and $C_j^c$ are as previously defined. If $\gamma_i = 1$, the cost function~\eqref{eq:J_i^v} will be similar to~\eqref{eq:J_i^s} for exit $j$, and $J_i^b$ would simply be equal to the cost of traversing the lanes that connect to exit $j$.  But, if $\gamma_i > 1$, the additional cost that the bypassing vehicles must pay due to traversing a longer path for joining their appropriate exit, is modeled. In fact, $\gamma_i > 1$ can model the cost incurred on bypassing vehicles due to the additional distance they need to traverse as well as the discomfort cost they will face for changing their lanes.


\begin{example}
Consider the diverge shown in Figure~\ref{fig:diverge}. In this case, $x_1^s$ is the fraction of the vehicles that remain on lane I and take exit 1, whereas $x_2^b$ is the fraction of the vehicles that use lane I and leave lane I close to the diverge to take the exit 2. In this case, $J_1^s = C_1^t \left( x_1^s + x_2^b \right) + C_1^c x_1^b \left ( x_1^s + x_2^b \right)$. Note that $C_1^t \left( x_1^s + x_2^a \right)$ is the cost of traversing lane I, where $ \left( x_1^s + x_2^a \right)$ is the total fraction of vehicles present on lane I. Now, consider $J_1^b$ which is the cost of the bypassing vehicles that take exit 1 by remaining on lane II until they are close to the diverge, and then change their lane to take the exit I. 
For this type of vehicles, $J_1^b = C_2^t \left(x_2^s + \gamma_1 x_1^b \right) + C_2^c x_2^b \left(  x_2^s + x_1^b \right)$. In this case, $C_2^t \left(x_2^s + x_1^b \right) + C_2^c x_2^b \left(  x_2^s + x_1^b \right)$ is the cost of traversing lane II (up to vicinity of the diverge); whilst, the additional cost $C_2^t ((\gamma_1 - 1)x_1^b)$ is due to traversing the extra distance required for leaving lane II and finally joining exit I, as well as the discomfort cost the vehicles have to pay for changing their lanes from II to I.
\end{example}

We let $\mathbf{C} = (C_i^t, C_i^c, \gamma_i: i \in I) $ be the vector of cost coefficients in our model.  Before proceeding, we need to introduce the following definition. 
\begin{definition} \label{def:mon}
A function $h(.): \mathbb{R}^n \longrightarrow \mathbb{R}$ is called elementwise monotone if and only if for every $\mathbf{x}, \mathbf{x}' \in \mathbb{R}^n$ such that $\mathbf{x} \leq \mathbf{x}'$, where inequalities are interpreted elementwise, we have
\begin{align*}
h(\mathbf{x}) \leq h(\mathbf{x}').
\end{align*}
\end{definition}

\noindent Using Equations~\eqref{eq:J_i^s} and~\eqref{eq:J_i^v}, the following remark is evident.
\begin{remark}\label{rem:mon}
For each $i \in I$, the cost functions $J_i^s$ and $J_i^b$, are elementwise monotone in the sense of Definition~\ref{def:mon}.
\end{remark}

\noindent We will later use Remark~\ref{rem:mon} to guarantee certain properties of our model.

A reasonable and realistic assumption is that vehicles act \emph{selfishly}, i.e. each vehicle acts in a manner that minimizes its own cost. We now assume that each vehicle has two options: either to choose its appropriate lane upstream of the diverge or to perform a bypass and take its target exit close to the diverge via a ``tight'' lane change. Therefore, we model the lane choice of vehicles at a traffic diverge as an equilibrium. Thus, at equilibrium, if for an exit $i \in I$, both $x_i^s$ and $x_i^b$ are nonzero, we must have $J_i^s(\mathbf{x}) = J_i^b(\mathbf{x})$ since otherwise, vehicles will move to the lanes with lower cost. If either $x_i^s$ or $x_i^b$ is zero, then, its corresponding cost must be already larger than the cost of the one with nonzero flow. These conditions are called the Wardrop conditions~\cite{wardrop1952some} in the transportation literature. In order to describe the formal definition of Wardrop conditions, let $G = (F, \bb{C})$ be a tuple enclosing $F$ and $\bb{C}$ which are respectively the normalized demand configuration and the vector of cost coefficients. 
Then, in our setting, an equilibrium is defined via the following.
\begin{definition}
For a given $G = ( F,\mathbf{C})$, a flow vector $\mathbf{x}$ is an equilibrium if and only if for every exit link $i \in I$, we have:
\begin{subequations}
\label{eq:eq_def}
\begin{gather}
\begin{align}
x_i^s (J_i^s(\mathbf{x}) - J_i^b(\mathbf{x})) &\leq 0 ,\\
x_i^b (J_i^b(\mathbf{x}) - J_i^s(\mathbf{x})) &\leq 0. 
\end{align}
\end{gather}
\end{subequations}
\end{definition}
\noindent Note that Equations~\eqref{eq:eq_def} imply that for an exit link $i \in I$, if $x_i^s \neq 0$ and $x_i^b \neq 0$, then at equilibrium, we must have $J_i^s(\mathbf{x}) = J_i^b(\mathbf{x})$. Alternatively, if at equilibrium $x_i^s = 0$ ($x_i^b = 0$) , we have $J_i^s(\mathbf{x}) \geq J_i^b(\mathbf{x})$ $\left( J_i^b(\mathbf{x}) \geq J_i^s(\mathbf{x})\right)$. Note that the adoption of a Wardrop assumption implies that vehicles can be
treated infinitesimally, i.e. the change caused by the unilateral lane change of a single vehicle is negligible. This is in accordance with our goal
of modeling the macroscopic behavior of vehicles at diverges


\section{Equilibrium Properties}\label{sec:eq_prop}

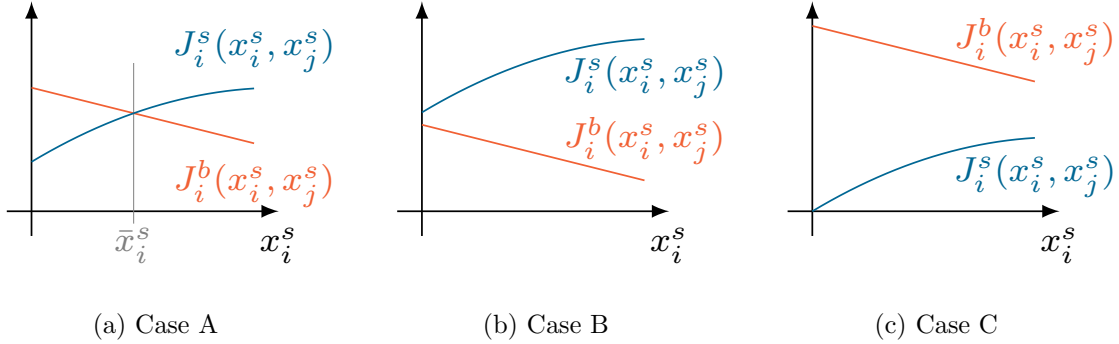
\begin{figure*}
	\centering
	\begin{subfigure}
    {0.3\textwidth} 
    \centering
\resizebox{6cm}{!}{
\begin{tikzpicture}[domain=0:1.8]
\scalebox{0.8}{
    \draw[->] (-0.2,0) -- (2,0) node[below] {\scriptsize $x_i^s$};
    \draw[->] (0,-0.2) -- (0,1.7) node[above] {};
  
     \draw[very thin,color=gray]  (0.826,-0.1) -- (0.826,1.2) node[below] {};
    \draw[color=RedOrange,variable=\x]  plot({\x},{-0.25*\x + 1}) 
        node[below] {\scriptsize $J_i^b(x_i^s, x_j^s)$};
    \draw[color=MidnightBlue,variable=\x] plot ({\x},{-0.15*(\x-2)*(\x-2)+1}) 
        node[above] {\scriptsize $J_i^s(x_i^s, x_j^s)$};
        
        \node[below,color=gray]  at (0.826,0) {\scriptsize $\bar{x}_i^s$}; 
}
\end{tikzpicture}}	
\caption{Case A} 
	\end{subfigure}
    \quad
    \begin{subfigure}{0.3\textwidth}
    \centering
    \resizebox{6cm}{!}{
\begin{tikzpicture}[domain=0:1.8]
\scalebox{0.8}{
    \draw[->] (-0.2,0) -- (2,0) node[below] {\scriptsize $x_i^s$};
    \draw[->] (0,-0.2) -- (0,1.7) node[above] {};
  
    \draw[color=RedOrange,variable=\x] plot ({\x},{-0.25*\x + 0.7} )
        node[above] {\scriptsize $J_i^b(x_i^s, x_j^s)$};
    \draw[color=MidnightBlue,variable=\x] plot ({\x},{-0.15*(\x-2)*(\x-2)+1.4}) 
        node[below] {\scriptsize$J_i^s(x_i^s, x_j^s)$};
}   
\end{tikzpicture}}
\caption{Case B}
    \end{subfigure}
    \quad
    	\begin{subfigure}{0.3\textwidth} 
    \centering
		\resizebox{6cm}{!}{
\begin{tikzpicture}[domain=0:1.8]
\scalebox{0.8}{
    \draw[->] (-0.2,0) -- (2,0) node[below] {\scriptsize $x_i^s$};
    \draw[->] (0,-0.2) -- (0,1.7) node[above] {};
  
    \draw[color =RedOrange,variable=\x] plot ({\x},{-0.25*\x + 1.5})
    node[above] {\scriptsize $J_i^b(x_i^s, x_j^s)$};
    
    \draw[color=MidnightBlue,variable=\x] plot ({\x},{-0.15*(\x-2)*(\x-2)+0.6} )
        node[below] {\scriptsize $J_i^s(x_i^s, x_j^s)$};
} 
\end{tikzpicture}}
		\caption{Case C} 
	\end{subfigure}
	\caption{Three possible configurations of $J_i^s(.)$ and $J_i^b(.)$.} 
    \label{fig:three_conf}
\end{figure*}

In this section, we state the properties of the equilibrium of our model including its existence and uniqueness.
\subsection{Equilibrium Existence}
Using the Existence Theorem in~\cite{braess1979existence} for the setting of our model, we can conclude that there always exists at least one equilibrium for a given $G = (F,\bb{C})$ if the following holds.

\begin{proposition}
Given $G = (F,\mathbf{C})$ for a traffic diverge, if the cost functions $J_i^s(\bb{x}), J_i^b(\bb{x}), i \in I$ are continuous and elementwise monotone in $\mathbf{x}$, then, there exists at least one Wardrop equilibrium for $G$. 
\end{proposition}
\begin{remark}

For a diverge with two exit links, using Remark~\ref{rem:mon} and continuity of $J_i^s(.)$ and $J_i^b(.)$, we can conclude that there always exists at least one equilibrium for every $G =(F,\mathbf{C}) $. 
\end{remark}

\subsection{Equilibrium Uniqueness}
Once the existence of Wardrop equilibrium is established, it is important to study its uniqueness. Equilibrium uniqueness is a desired and favorable property of a model since the social cost at equilibrium is well defined when there exits only one equilibrium. In this subsection, we show that our model has this favorable property. Equations~\eqref{eq:J_i^s} and~\eqref{eq:J_i^v} indicate that $J_i^s(.)$ depends not only on $x_i^s$ but also on $x_i^b$ and $x_j^b$. In the routing games literature, this dependence is referred to as the cost functions being \emph{nonseparable}~\cite{correa2008geometric}. This nonseparability is further asymmetric, meaning that the incurred costs are not the same across $J_i^s$ and $J_i^b$ for each $i \in I$. In addition to their asymmetric nonseparability, the cost functions in~\eqref{eq:J_i^s} and~\eqref{eq:J_i^v} are nonlinear. It is known that generally under such settings, Wardrop equilibria exhibit very complicated behavior including nonuniqueness. Equilibrium uniqueness is generally only achieved under very strong assumptions, which do not hold in the majority of applications~\cite{altman2006survey}. Despite this complication, and the fact that none of the existing results in the literature on sufficient conditions for the uniqueness of an equilibrium can be applied to our model, we are able to obtain the conditions under which a given $G = (F,\mathbf{C})$ is guaranteed to have a unique equilibrium in our model. 

To prove uniqueness, we first define an auxiliary game such that there exists a connection between the Wardrop equilibrium in our model and Nash equilibrium of the related auxiliary game. For any given $G = (F,\mathbf{C})$, we define a two player game $\tilde{G} = \langle P, A, (\tilde{J}_p: p \in P) \rangle$, where $P = \{ 1,2\}$ is the set of players. Since both $I$ and $P$ are the set $\{1,2\}$, we use a bijective correspondence between every $p \in P$ and $i \in I$. In fact, $p=1$ ($p=2$) implies that $i = 1$ ($i=2$) and vice versa. Therefore, every exit link $i \in I$ is associated with a player $p \in P$. In the auxiliary game that will be defined subsequently, $A = A_1 \times A_2$ is the action space, where for each $p\in P$, $A_p = [0,f_p] = [0, f_i]$ is the action set of player $p$. Moreover, $\tilde{J}_p$ is the cost associated with each player $p \in P$. We let $\bb{y} = (y_p, p \in P)$ be the vector of actions taken by the two players of the game $\tilde{G}$. To make a  connection between our traffic diverge setting and the defined auxiliary game, for a given flow vector $\mathbf{x}$, we define $\bb{y} $ to be
\begin{align}
\bb{y} = (x_i^s, i \in I).
\end{align}
\noindent Then, for every $ p \in P$, we define $\tilde{J}_p(\bb{y})$ to be
\begin{align}\label{eq:nash_cost}
\tilde{J}_p(\mathbf{y}) = \left( J_i^s(\bb{x}) - J_i^b(\bb{x})\right)^2.
\end{align}
In the auxiliary game $\tilde{G}$, a vector $\bb{y} = (y_p, y_{p'})$ is a pure Nash equilibrium if and only if
\begin{align}\label{eq:nash}
\forall\, p, p' \in P,\; y_p &= B_p(y_{p'})\ \\
&= \text{argmin}_{y_p \in [0,f_p]} \tilde{J}_p(y_p, y_{p'})
\end{align}
where $B_p$ is the best response function of player $p$. Note that since for every player $p\in P$, $\tilde{J}(\bb{y})$ is a continuous function on a closed interval, a minimum is achieved. Equation~\eqref{eq:nash} implies that if $y_{p'}$ is fixed, player $p$ takes the best possible action that minimizes its own cost $\tilde{J}_p(\bb{y})$. The following proposition establishes the connection between the Wardrop equilibrium of $G$ and the Nash equilibrium of $\tilde{G}$.

\begin{proposition}\label{prop:two-person}
A flow vector $\mathbf{x} = (x_i^s, x_i^b: i \in I)$ is a Wardrop equilibrium for $G = (F,\bb{C})$ if and only if $\bb{y}= (x^s_i, i \in I)$ is a pure Nash equilibrium for $\tilde{G}$ provided that
\begin{align}\label{eq:incr_dec}
C_i^t \geq C_i^c,\; \forall i \in I.
\end{align}
\end{proposition}
\begin{proof}
First note that given the normalized demand configuration $F = \{f_1, f_2\}$, flow conservation requires that for every exit link $i \in I$, we have $x_i^b = f_i - x_i^s$. Thus, with a little abuse of notation, $J_i^s(\bb{x})$ and $J_i^b(\bb{x})$ can be written as $J_i^s(x_i^s, x_j^s)$ and $J_i^b(x_i^s, x_j^s)$ for every pair of exit links $i \neq j \in I$. For a given $x_j^s$, we show that~\eqref{eq:incr_dec} is a sufficient condition for $J_i^s(x^s_i, x^s_j)$ to be increasing in $x^s_i$, and $J_i^b(x^s_i, x^s_j)$ be decreasing in $x^s_i$ . To see this, note that for every $i \neq j \in I$, we have:
\begin{align}\label{eq:jll_der}
\frac{\partial J^s_i (x_i^s, x_j^s)}{\partial x_i^s} = -2 C_i^c x_i^s + C_i^t + C_i^c x^s_i - C_i^c (f_j - x_j^s).
\end{align}
Equation~\eqref{eq:jll_der} is linear in $x^s_i$. Moreover, for each $i \in I$, $x^s_i$ is allowed to only take values in interval $[0,f_i]$. Therefore, in order to obtain sufficient conditions for the positivity of~\eqref{eq:jll_der}, it is sufficient to guarantee that $\frac{\partial J^s_i (x_i^s, x_j^s)}{\partial x_i^s}$ is positive at all possible extreme points $(x^s_i,x^s_j )$ which are $\{\left( 0, 0\right), \left( f_1,  0 \right), \left( 0, f_2\right),\left( f_1, f_2\right) \}$. Using the fact that the demand functions must satisfy $f_1 + f_2 = 1$, it is easy to verify that the smallest possible value of~\eqref{eq:jll_der} is attained in $(f_1,0)$ when $f_1 = 1$. At the point $(1,0)$, we have $\frac{\partial J^s_i}{\partial x_i^s}(1,0) = C_i^t - C_i^c$. Therefore,~\eqref{eq:incr_dec} is a sufficient condition for $J_i^s (x_i^s, x_j^s)$ to be increasing in $x_i^s$. Similarly, we can compute $\frac{\partial J^b_i (x_i^s, x_j^s)}{\partial x_i^s}$ which is
\begin{align}
\frac{\partial J^b_i (x_i^s, x_j^s)}{\partial x_i^s} &= -C^t_{j}\gamma_j - C_j^c(f_j-x_j^s).
\end{align}
Since $(f_j-x_j^s)$ is always greater than or equal to zero, clearly, for every $i\neq j \in I$, $J_i^b (x_i^s, x_j^s)$ is always decreasing in $x_i^s$ for any given $ x_j^s$.

Now, we can proceed to proving that under~\eqref{eq:incr_dec}, every Wardrop equilibrium of $G$ is equivalent to a Nash equilibrium of the auxiliary game $\tilde{G}$. Consider best response function $B_p(y_{p'})$ in~\eqref{eq:nash}. For a given $y_{p'}=x_j^s$, in order to minimize $\tilde{J}_p(y_p, y_{p'})$ over $y_p = x_i^s$, since $ J^s_i$ is increasing, and $ J^b_i$  is decreasing in $x_i^s$ under~\eqref{eq:incr_dec}, the following three scenarios may occur for a given $x_j^s$ (see Figure~\ref{fig:three_conf}):

\begin{itemize}
\item Case A: $J_i^s(x_i^s, x_j^s)$ and $J_i^b(x_i^s, x_j^s)$ have an intersection  on the interval $(0,f_p)$. In this case, there exists a point ${\bar{x}^s}_i(x_j^s) \in (0, f_p)$ such that $J_i^s(\bar{x}_i^s, x_j^s) = J_i^b(\bar{x}_i^s, x_j^s)$ (See Figure~\ref{fig:three_conf}, case A). Using~\eqref{eq:nash_cost}, it can also be verified that in this case, the intersection point $y_p = \bar{x}^s_i$ is the best response for a given $y_{p'} = x _j^s$. If this is the case, Equations~\eqref{eq:eq_def} are also satisfied by $\bar{x}_i^s$ for a given $x_j^s$. It is easy to see that the reverse is also true. Indeed, if $\bar{x}_{i}^s \in (0,f_i)$ satisfies~\eqref{eq:eq_def} for a given $x_{j}^s$, then, $\bar{x}_{i}^s$ must be the intersection of $J_i^s(x_i^s, x_{j}^s)$ and $J_i^b(x_i^s, x_{j}^s)$ on the interval $(0,f_i)$. Therefore, $y_p = x_i^s$ is the best response of $y_{p'} = x_j^s$.

\item Case B: $J_i^s(x_i^s, x_j^s)$ and $J_i^b(x_i^s, x_j^s)$ do not intersect on the interval $(0,f_p)$, and $J_i^s(0, x_j^s) \geq J_i^b(0, x_j^s)$ for a given $x_j^s$. In this case, if $y_{p'} = x_j^s$, then ${y}_p = B_p(y_{p'}) = 0$ (See Figure~\ref{fig:three_conf}, case B). It is easy to see that, $x_i^s = 0$ satisfies~\eqref{eq:eq_def} for a given $x_{j}^s$ since if $x_i^s = 0$, then, $x_i^b = f_i$ while $J^s_i \geq J^b_i$. The reverse is also true, if $x_i^s = 0$ satisfies~\eqref{eq:eq_def} for a given $x_j^s$, then $y_p = x_p^s = 0$ is the best response of ${y}_{p'} = x_j^s$.

\item Case C: $J_i^s(x_i^s, x_j^s)$ and $J_i^b(x_i^s, x_j^s)$ do not intersect on the interval $(0,f_p)$, and $J_i^s(0, x_j^s) \leq J_i^b(0, x_j^s)$. In this case, if $y_{p'} = x_j^s$, then $y_p = B_p(y_{p'}) = 1$. Similar to case B, one can conclude that if $y_{p'} = x_j^s$, then $y_p = {x}_i^s = 1$ is equal to $B_p(y_{p'})$ if and only if ${x}_i^s= 1$ satisfies~\eqref{eq:eq_def} for a given $x_j^s$.
\end{itemize}
So far, we have shown that for every
$p\neq p' \in P$, for a given $y_{p'}$, $y_p$ is the best response of $y_{p'}$ if and only if $\mathbf{x} = (y_p, f_p-y_p: i \in I)$ satisfies~\eqref{eq:eq_def}. Therefore, $\bb{y} = (x_i^s, i \in I)$ is a Nash equilibrium of $\tilde{G}$ if and only if $\mathbf{x} = (x_i^s, f_i-x_i^s)_{i \in I}$ is a Wardrop equilibrium of $G$.
\end{proof}

\begin{remark}\label{rem:eq_inter}
Notice that using the three cases described in the proof of Proposition~\ref{prop:two-person}, for a given $y_{p'} = x_j^s$, the best response $B_p(y_{p'})$ can be found by first intersecting $J_i^s(x_i^s, x_j^s)$ and $J_i^b(x_i^s, x_j^s)$ and then projecting the intersection point $\bar{x}^s_i(x_j^s)$ onto the interval $[0,f_i]$. We will use this fact in the remainder to prove equilibrium uniqueness. 
\end{remark}

Having Proposition~\ref{prop:two-person} in mind, we are ready to state and prove the following.
\begin{theorem}\label{the:uniq}
For a given diverge $G = (F, \bb{C})$, a Wardrop equilibrium flow vector $\mathbf{x}$ is unique if
\begin{align}
C_i^t &\geq C_i^c,\; \forall i \in I, \label{eq:the_nash_cond}\\
(\gamma_i - 1) C_j^t &\geq C_i^c,\; \forall i \in I \label{eq:small_der}.
\end{align}
\end{theorem}
\begin{proof}
Construct the auxiliary game $\tilde{G} = \langle P, A, (\tilde{J}_p, p \in P) \rangle$ described above. Using Proposition~\ref{prop:two-person}, we know that if~\eqref{eq:the_nash_cond} holds, $\bb{x}$ is a Wardrop equilibrium for $G$ if and only if $(y_p:p \in P) = (x^s_i : i \in I)$ is a Nash equilibrium for $\tilde{G}$. We now prove that under~\eqref{eq:small_der}, $\tilde{G}$ has a unique equilibrium; thus, $G$
 must also have a unique equilibrium. To see this, note that $\bb{y} = (x_i^s: i \in I)$ is a Nash equilibrium for $\tilde{G}$ if and only if for every $p \neq p'$, $y_p = B_p(y_{p'}),$ and  $y_{p'} = B_{p'}(y_p)$. 
 These conditions can be rewritten as
 \begin{subequations}
\label{eq:eq_cond_casc}
\begin{gather}
\begin{align}
 y_p &= B_p(B_{p'}(y_p)), \\
 y_{p'} &= B_{p'}(B_{p}(y_{p'})).
\end{align}
\end{gather}
\end{subequations}
Equations~\eqref{eq:eq_cond_casc} indicate that $\bb{y}$ is an equilibrium if and only if for every $p \neq p' \in P$, $y_p$ is a fixed point for $B_p\left(B_{p'}(.)\right)$. Thereby, $(y_p, y_{p'})$ is an equilibrium for $\tilde{G}$ if and only if $B_p(B_{p'}(.))$ intersects the line going through the origin with slope 1 at $y_p$, and $B_{p'}(B_p(.))$ intersects the line going through the origin with slope 1 at $y_{p'}$. In the remainder, we prove that under~\eqref{eq:small_der}, the slope of $B_p(B_{p'}(.))$ is always positive and smaller than 1 for every $p \neq p' \in P$. Therefore, $B_p(B_{p'}(.))$ can intersect the identity line at most once. Thus, knowing that there exits at least one equilibrium, we can conclude that $\tilde{G}$ and therefore $G$ always has a unique equilibrium if~\eqref{eq:the_nash_cond} and~\eqref{eq:small_der} hold. To prove this, it suffices to show that $ 0 \leq\frac{d B_p (y_{p'})}{d y_{p'}} \leq 1$, for every $p\neq p' \in P$. 
To see this, let $x_j^s$ be such that 
$J_i^s(x_i^s, x_j^s)$ and $J_i^b(x_i^s, x_j^s)$ intersect each other at $\bar{x}_i^s(x_j^s) \in [0,f_i]$. Using~\eqref{eq:J_i^s} and the fact that for every $ i \in I$, $x_i^b = f_i - x_i^s$, we have
\begin{align}\label{eq:11}
\frac{\partial J_i^s (x_i^s, x_j^s)}{\partial x_j^s} &= -C^t_i - C^c_i (f_i - x_i).
\end{align}
Since $(f_i - x_i) \geq 0$, we can conclude that $\frac{\partial J_i^s (x_i^s, x_j^s)}{\partial x_j^s} \leq 0$.
Similarly, we can compute
\begin{align}\label{eq:22}
\frac{\partial J_i^b (x_i^s, x_j^s)}{\partial x_j^s} &= C^t_j + C_j^c (f_j - x_j^s) - C_j^c(x_j^s + f_i - x_i^s).
\end{align}
Since $(f_j - x_j^s) \geq 0$, it is easy to see that if~\eqref{eq:the_nash_cond} holds, $\frac{\partial J_i^b (x_i^s, x_j^s)}{\partial x_j^s}$ is always positive. Thus,~\eqref{eq:11} and~\eqref{eq:22} imply that as $x_j^s$ increases, $J_i^s$ decreases while $J_i^b$ increases. Therefore, as $x_j^s$ increases, $\bar{x}_i^s (x_j^s)$ can only increase. However, Remark~\ref{rem:eq_inter} implies that if $\bar{x}_i^s (x_j^s)$ lies outside the interval $[0,f_i]$, it is projected on this interval. Since $x_j$ varies on interval $[0,f_j]$, interval $[0,f_j]$ can be divided into $[0,f_j] = [0,m_j] \cup [m_j,n_j] \cup [n_j,f_j]$, such that $\bar{x}_i^s (x_j^s)$ is always 0 for $x_j^s \in [0,m_j]$, and always 1 on for $x_j^s \in [n_j,f_j]$. Note that either of the intervals $[0,m_j]$, $[m_j,n_j]$ and $[n_j,f_j]$ can possibly be empty. Hence, in order to show that the slope of the best response function $B_p(x_j^s)$ is always smaller than $1$, it suffices to prove that it is indeed less than 1 for $x_j^s$ being in interval $[m_j,n_j]$ where $J_i^s(x_i^s, x_j^s)$ and $J_i^b(x_i^s, x_j^s)$ do intersect at $\bar{x}_i^s \in [0,f_i]$.



For a given $x_j^s \in [m_j,n_j]$, $\bar{x}_i^s(x_j^s)$ must satisfy
\begin{align*}
J_i^s(\bar{x}_i^s, {x}_j^s) - J_i^b(\bar{x}_i^s, {x}_j^s) =0.
\end{align*}
Therefore, using implicit differentiation, $\frac{ d\bar{x}_i(x_j^s)}{d x_j}$ can be computed via

\begin{equation}\label{eq:impl_diff}
\frac{\partial}{\partial x_i^s }\left( J_i^s(\bar{x}_i^s,{x}_j^s) - J_i^b(\bar{x}_i^s,{x}_j^s) \right) \frac{ d\bar{x}^s_i(x_j^s)}{d x^s_j} + \frac{\partial}{\partial x_j^s}\left( J_i^s(\bar{x}_i^s,{x}_j^s) - J_i^b(\bar{x}_i^s,{x}_j^s) \right) = 0. 
\end{equation}


\noindent Using~\eqref{eq:J_i^s} and~\eqref{eq:J_i^v} , and the fact that $x_i^b = f_i - x_i^s$ for all exit links $i \in I$, we have
\begin{equation}\label{eq:diff}
\frac{\partial}{\partial x_j^s}\left( J_i^s(\bar{x}_i^s,{x}_j^s) - J_i^b(\bar{x}_i^s,{x}_j^s) \right) = -C_i^t - C_i^c (f_i - x_i^s)-
C_i^c + C_j^c (x_j^s + f_i - x_i^s) - C_j^c (f_j - x_j^s).
\end{equation}

\noindent Since~\eqref{eq:diff} is linear in $x_i^s$ and $x_j^s$, its maximum and minimum are attained at its extreme points. It is easy to check that the maximum possible value for~\eqref{eq:diff} is $-C^t_i - C^t_j + C_j^c$. If~\eqref{eq:the_nash_cond} holds, $-C^t_i - C^t_j + C_j^c \leq 0$. Therefore $\frac{\partial}{\partial x_j^s}\left( J_i^s(\bar{x}_i^s,{x}_j^s) - J_i^b(\bar{x}_i^s,{x}_j^s) \right) \leq 0$ under~\eqref{eq:the_nash_cond}. Using the same approach, one can verify that under~\eqref{eq:the_nash_cond}, it always the case that for every exit link $i \in I$,
 \begin{align*}
 \frac{\partial}{\partial x_i^s} \left( J_i^s(\bar{x}_i^s,{x}_j^s) - J_i^b(\bar{x}_i^s,{x}_j^s) \right)  \geq 0.
 \end{align*}
 Hence, using~\eqref{eq:impl_diff}, under~\eqref{eq:the_nash_cond}, 
 $$\frac{ d\bar{x}^s_i(x_j^s)}{d x_j^s} \geq 0, \quad \forall i\neq j \in I.$$
 
\begin{figure*}
    \centering
    \begin{subfigure}[b]{0.48\textwidth}
       \centering \includegraphics[width=\textwidth]{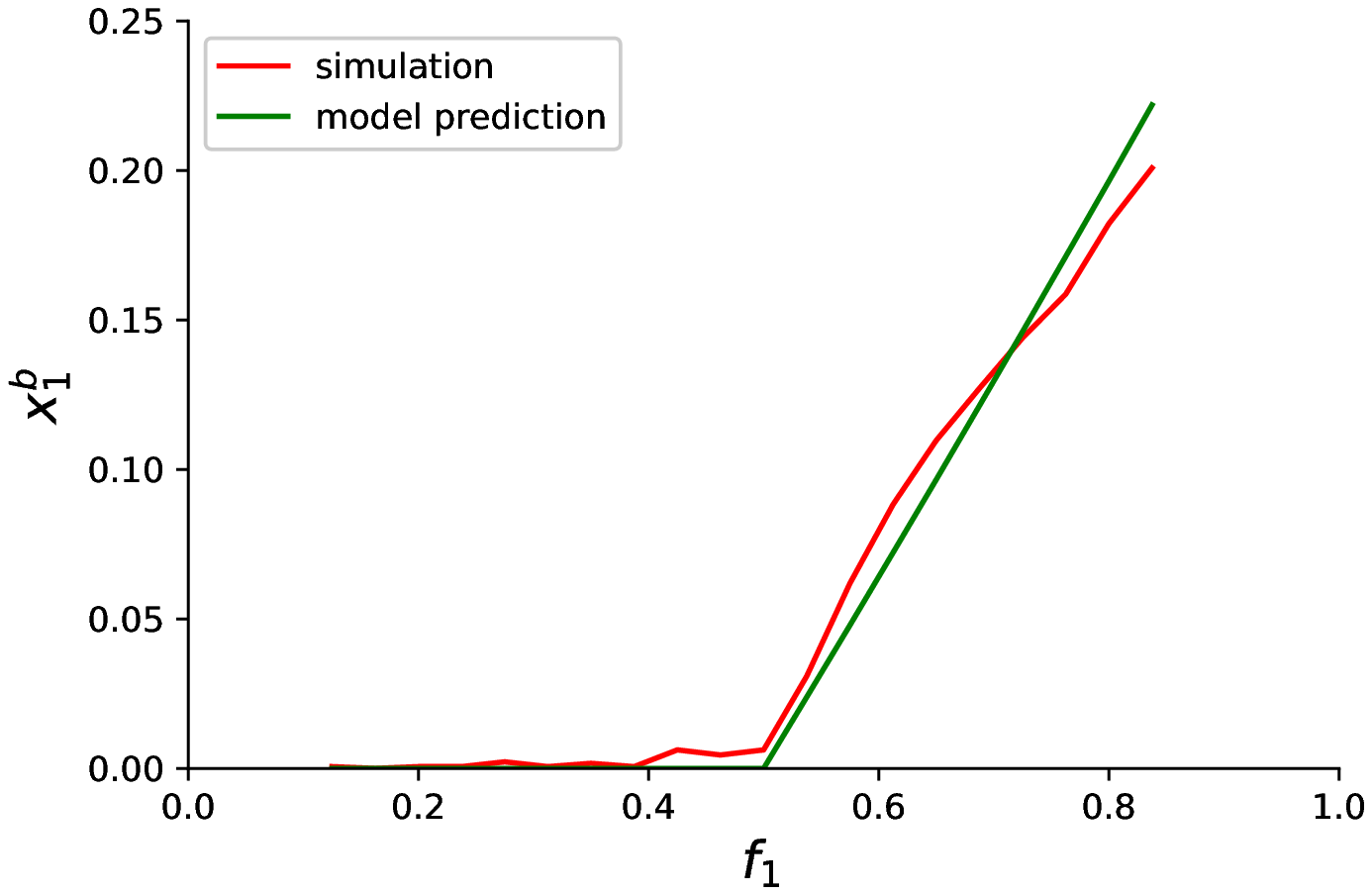}
        \caption{$x_1^b$}
    \end{subfigure}\quad 
    ~ 
      \begin{subfigure}[b]{0.48\textwidth}
       \centering \includegraphics[width=\textwidth]{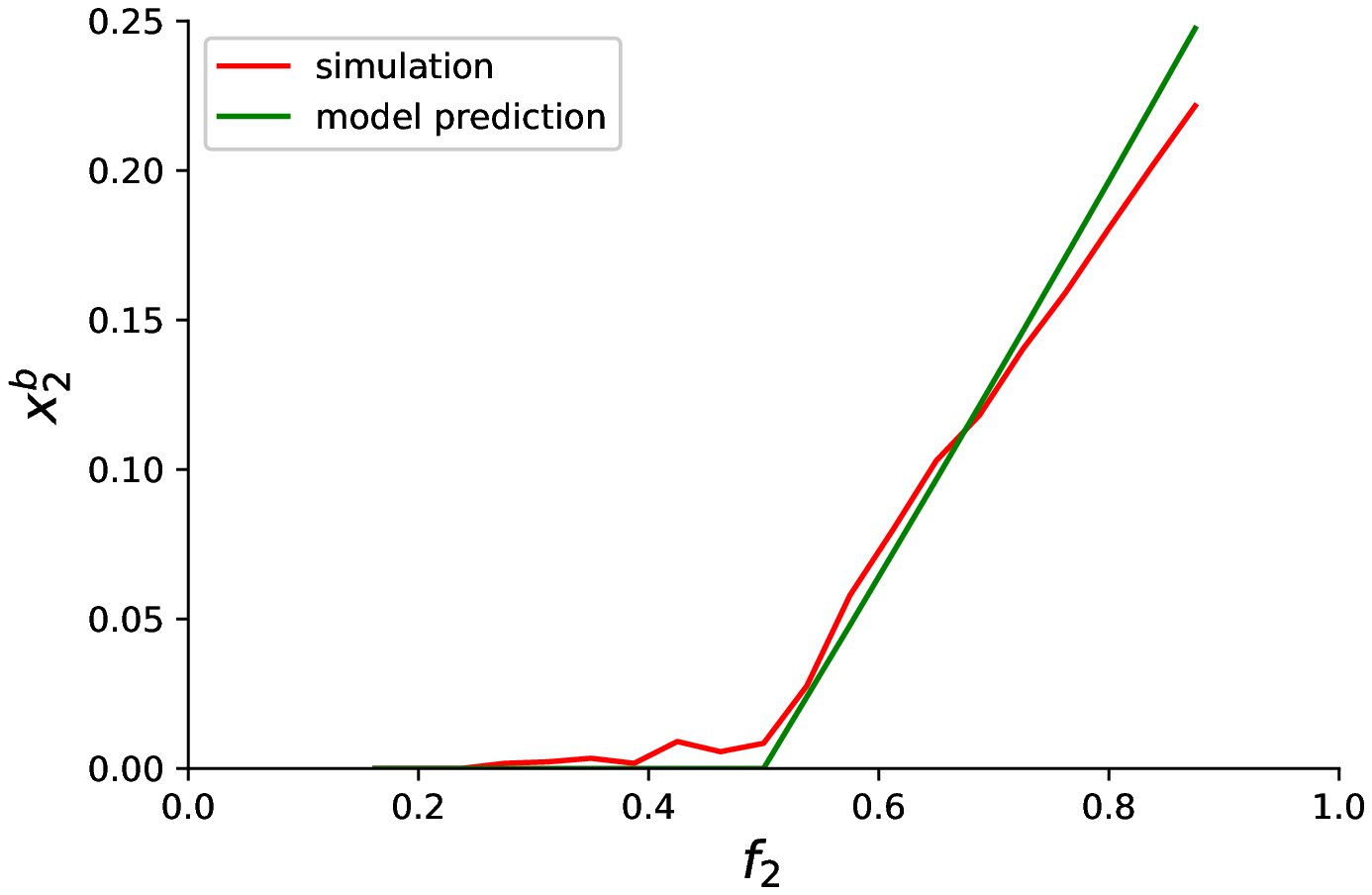}
        \caption{$x_2^b$}
    \end{subfigure}
    \caption{The fraction of bypassing vehicles, $x_i^b$, predicted by our macroscopic model and the values measured from microscopic simulations as a function of demand fractions $f_i$'s.}\label{fig:model_pred}
\end{figure*}
 
\noindent Now that we have shown that the slope of the best response function is always positive, it only remains to prove that $\frac{ d\bar{x}^s_i(x_j^s)}{d x_j^s} \leq 1$. To prove this, it suffices to show that
\begin{equation}\label{eq:diff_eval}
\frac{\partial}{\partial x_i^s} ( J_i^s(\bar{x}_i^s,{x}_j^s) - J_i^b(\bar{x}_i,{x}_j^s) ) \geq -\left(\frac{\partial}{\partial x_j^s}\left( J_i^s(\bar{x}_i^s,{x}_i^s) - J_i^b(\bar{x}_i^s,{x}_j^s) \right) \right).
\end{equation}

\noindent Substituting~\eqref{eq:J_i^s},~\eqref{eq:J_i^v}, and~\eqref{eq:diff} in~\eqref{eq:diff_eval} and computing the linear function at its extreme points, we observe that~\eqref{eq:small_der}
is a sufficient condition for~\eqref{eq:diff_eval} which completes our proof.
\end{proof}


\section{Simulation Results} \label{sec:model_val}


Up to now, we have described our model and its properties. In this section, we describe how our simulation results indicate that our model can successfully predict the bypassing behavior of vehicles. A key element of our model which affects its functionality is the coefficient vector $\bf{C}$. Therefore, in order to study the performance of the model, it needs to be calibrated first, i.e. the coefficient vector $\bf{C}$ that best fits a given diverge must be determined.

\subsection{Model Calibration}
Consider a diverge with two exit links $I = \{1,2 \}$. Fix the total flow of vehicles $d = d_1 + d_2$ that enter the diverge. For a given fixed $d$, consider different normalized demand configurations, $F^k = \{f_1^k, f_2^k \},1 \leq k \leq K$, where $K$ is the total number of possible normalized demand configurations available from the data or simulation. For each value of $f_1^k$ and $f_2^k = 1 - f_1^k$, measure $(x_i^s)^k$ and $(x_i^b)^k$ either from real data or simulation, where $(x_i^s)^k$ and $(x_i^b)^k$ are the fractions of steadfast and bypassing vehicles for each destination $i \in I$ when $k^{\text{th}}$ demand pattern is used. We let $\bb{x}^k$ represent the vector $\bb{x}$ measured for the $k^{\text{th}}$ demand configuration. Using our model, the vector of cost coefficients $\bb{C}$ must be determined such that Equations~\eqref{eq:eq_def} are satisfied by $(x_i^s)^k$ and $(x_i^b)^k$ for every $k \leq K$. But, since~\eqref{eq:eq_def} contains nonlinear inequalities, finding such a $\bf{C}$ is nontrivial. We propose the following calibration process for the cost parameters vector $\bf{C}$.

For every $k \leq K$ and $i \in I$, define the integer variables $(z_i^s)^k \in \{0,1 \}$, and $(z_i^b)^k \in \{0,1\}$ such that

\begin{subequations}
\label{eq:impl_cons}
\begin{gather}
\begin{align}
(x_i^s)^k (J_i^s(\mathbf{x}^k) - J_i^b(\mathbf{x}^k)) \leq 0 &\Longleftrightarrow (z_i^s)^k = 0 \\
(x_i^s)^k (J_i^s(\mathbf{x}^k) - J_i^b(\mathbf{x}^k)) > 0 &\Longleftrightarrow (z_i^s)^k = 1 \\
(x_i^b)^k (J_i^b(\mathbf{x}^k) - J_i^s(\mathbf{x}^k)) \leq 0 &\Longleftrightarrow (z_i^b)^k = 0 \\
(x_i^b)^k (J_i^b(\mathbf{x}^k) - J_i^s(\mathbf{x}^k)) > 0 &\Longleftrightarrow (z_i^b)^k = 1
\end{align} 
\end{gather}
\end{subequations}

    ~ 
    ~ 

\noindent Then, letting $\mathbf{z}$ be the vector of $(z_i^s)^k$ and $(z_i^b)^k$ for all $k$'s and all exit links $i \in I$, we propose to solve the following optimization problem for calibrating $\bb{C}$.

\begin{equation}\label{eq:cal_opt}
\begin{aligned}
& \underset{\bf{C}, \bf{z}}{\text{minimize}}
& & \sum_{k \in K} \sum_{i \in I} \left((z_i^s)^k + (z_i^b)^k \right) \\
& \text{subject to}
& & \text{Equations}~\eqref{eq:impl_cons} \\
&&& \bb{C}_r \geq 1,
\end{aligned}
\end{equation}
where $\bb{C}_r$ is the $r_{\text{th}}$ element of $\bb{C}$. We use the constraint $\bb{C}_r \geq r$ to avoid the optimizer from setting all the elements of $\bb{C}$ to be zero. It is important to note that since in~\eqref{eq:J_i^s} and~\eqref{eq:J_i^v}, every term is multiplied by one and only one element of $\bf{C}$, and multiplying all cost functions by the same constant does not change the Wardrop conditions, scaling $\bb{C}$ by a single number will not affect the model. Therefore, this constraint does not affect the model. Note that for every inequality constraint that is violated in~\eqref{eq:cal_opt}, the cost is increased by 1. Thus,~\eqref{eq:cal_opt} penalizes for not satisfying~\eqref{eq:eq_def} which are the equilibrium conditions. But, how can the optimization problem~\eqref{eq:cal_opt} be solved where the constraints are of the form~\eqref{eq:impl_cons}? To answer this, we use the procedure introduced in~\cite{raman2014model}. Let $M$ be a large positive number, and $\epsilon$ be a small positive number close to zero. For every $k$, the following is equivalent to~\eqref{eq:impl_cons}.

\begin{subequations}\label{eq:simp_impl}
\begin{gather}
\begin{align}
(x_i^s)^k (J_i^s(\mathbf{x}^k) - J_i^b(\mathbf{x}^k)) &\leq M (z_i^s)^k - \epsilon,  \\
-(x_i^s)^k (J_i^s(\mathbf{x}^k) - J_i^b(\mathbf{x}^k)) &\leq M (1-z_i^s) - \epsilon,  \\
(x_i^b)^k (J_i^b(\mathbf{x}^k) - J_i^s(\mathbf{x}^k)) &\leq M (z_i^b)^k - \epsilon, \\
-(x_i^b)^k (J_i^b(\mathbf{x}^k) - J_i^s(\mathbf{x}^k)) &\leq M (1-z_i^b) - \epsilon.
\end{align} 
\end{gather}
\end{subequations}
Therefore, our model can be calibrated by solving the following optimization problem
\begin{equation}\label{eq:cal_opt_final}
\begin{aligned}
& \underset{\bf{C}}{\text{minimize}}
& & \sum_k \sum_{i \in I} \left((z_i^s)^k + (z_i^b)^k \right) \\
& \text{subject to}
& & \text{Equations}~\eqref{eq:simp_impl} \\
&&& \bb{C}_j \geq 1.
\end{aligned}
\end{equation}
Note that~\eqref{eq:cal_opt_final} is now a mixed--integer linear program that can be easily solved using optimization packages. Since~\eqref{eq:cal_opt_final} is solved offline, and, further, the number of required integer variables is small, the computational cost for solving~\eqref{eq:cal_opt_final} is not overtaxing to calibrate our model.  

\subsection{Model Validation}
Consider the diverge shown in Figure~\ref{fig:diverge}. We used the microscopic traffic simulator SUMO~\cite{SUMO2012} to simulate the traffic behavior at the diverge of Figure~\ref{fig:diverge} for different demand configurations. 
A total flow of $d = 3000 \frac{\text{veh}}{\text{hour}}$ enters the diverge. The capacity of each lane is $930 \frac{\text{veh}}{\text{hour}}$. At every simulation, a fraction of vehicles $f_1$ is assumed to take the exit link 1 while the remaining fraction of vehicles $f_2 = 1-f_1$ is assumed to take the exit link 2. For different values of $f_1$, $x_1^s$, $x_1^b$, $x_2^s$, and $x_2^b$ are measured. Then, this data set is used to calibrate the model, i.e. finding the cost parameter vector $\bb{C}$ that best fits the data by solving optimization problem~\eqref{eq:cal_opt_final}. Since our road geometry is symmetric, we introduced the additional constraints that $C_1^t = C_2^t$, $C_1^c = C_2^c$, and $\gamma_1 = \gamma_2$ in~\eqref{eq:cal_opt_final}, and obtained the following values for $\bb{C}$.
\begin{align*}
C_1^t = C_2^t = 1, \; C_1^c = C_2^c = 1,\; \gamma_1 = \gamma_2 = 2.7.
\end{align*}
Notice that the obtained values of $\bb{C}$ satisfy~\eqref{eq:the_nash_cond} and~\eqref{eq:small_der}; thus, in every scenario, Theorem~\ref{the:uniq} implies that there exits only one equilibrium. We used a total of 20 fractional demand configurations; thus, a total of 40 integer variables.
The objective function of~\eqref{eq:cal_opt_final} was 4 when fitting $\bb{C}$, meaning that only 4 inequalities were unsatisfied among the 40 inequality constraints of our data set.

\begin{figure*}
    \centering
    \begin{subfigure}[b]{0.48\textwidth}
       \centering \includegraphics[width=\textwidth]{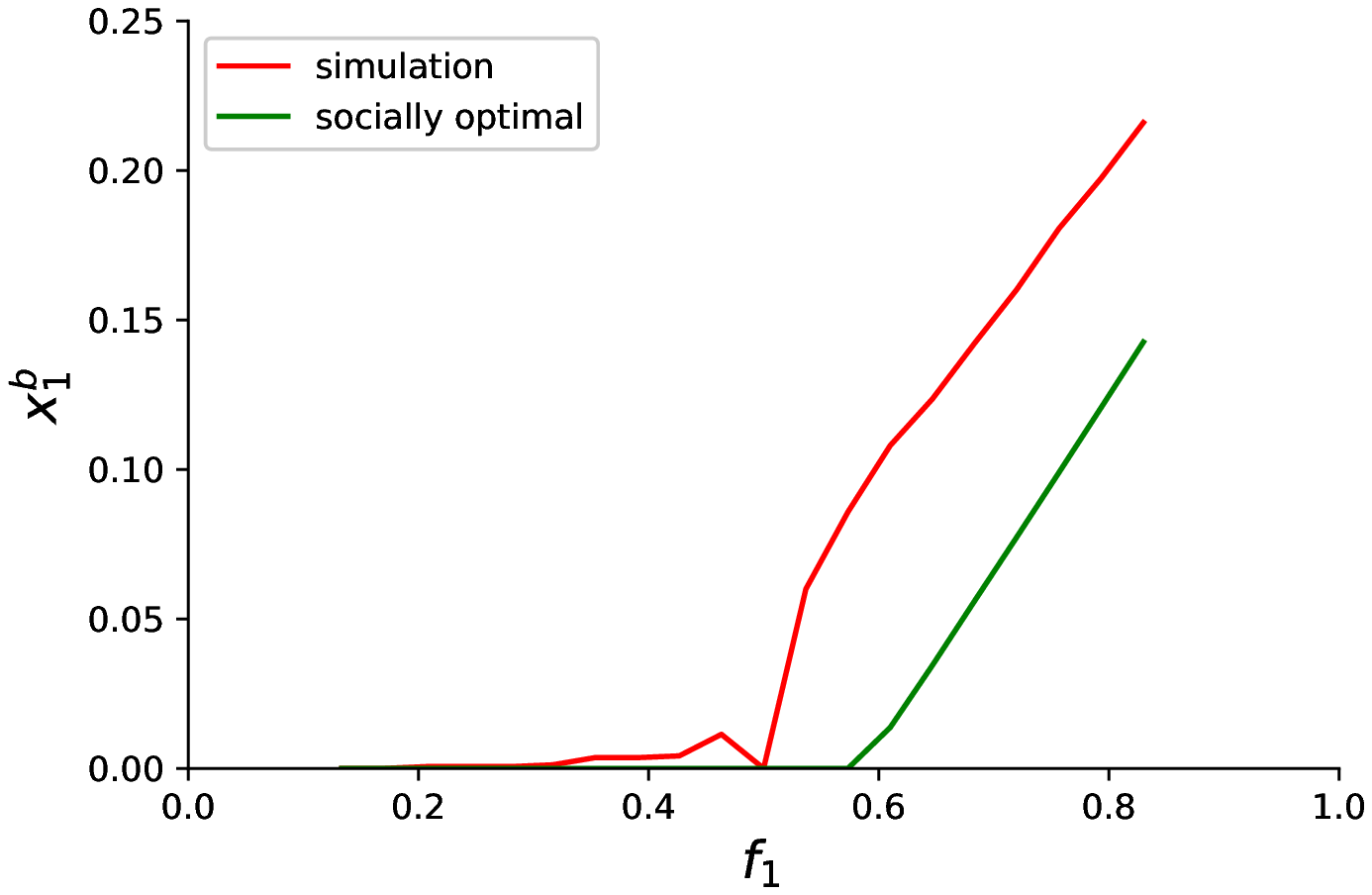}
        \caption{$x_1^b$}
    \end{subfigure}\quad
    ~ 
      \begin{subfigure}[b]{0.48\textwidth}
       \centering \includegraphics[width=\textwidth]{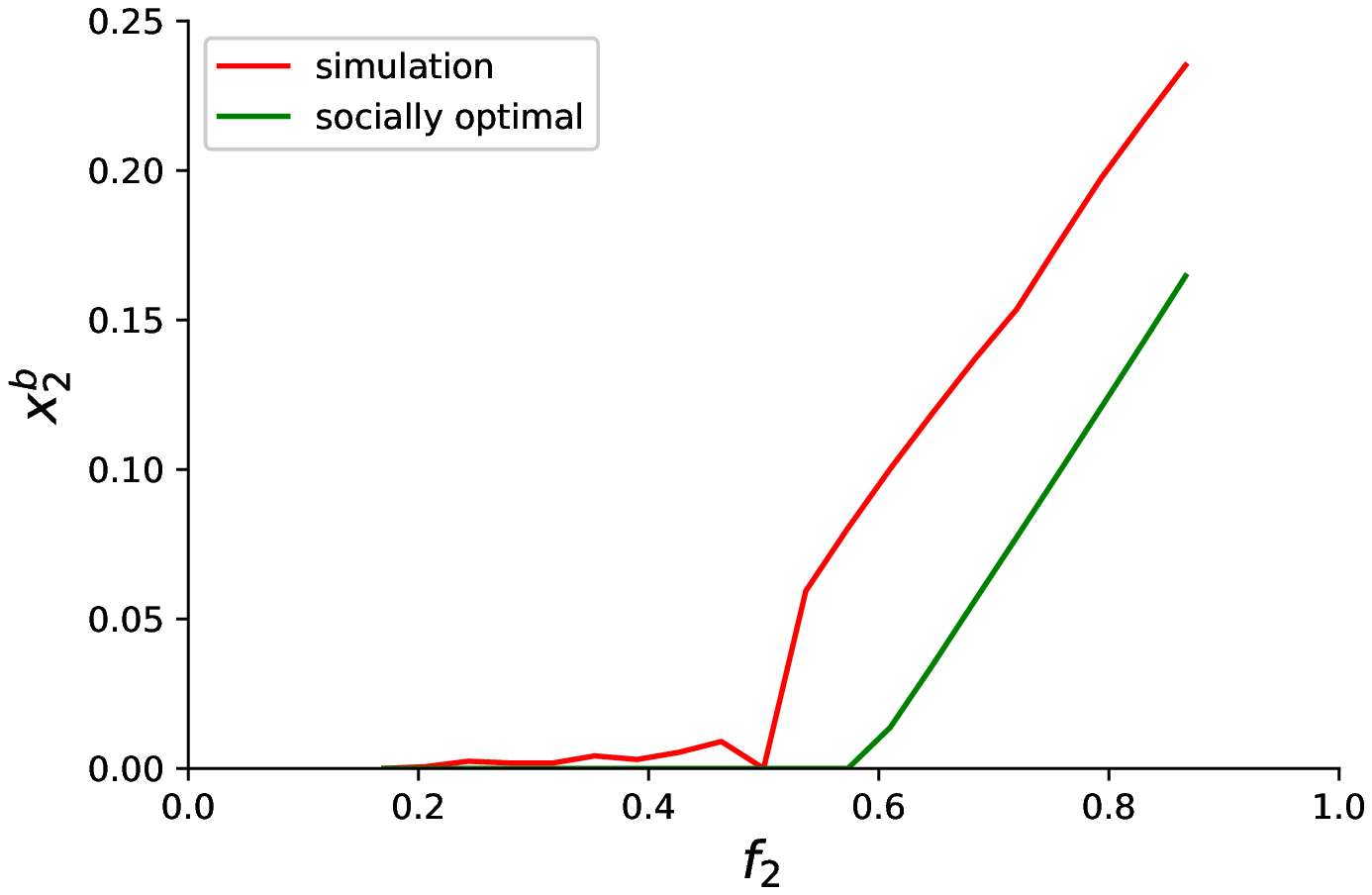}
        \caption{$x_2^b$}
    \end{subfigure}    
    \caption{The fraction of bypassing vehicles, $x_i^b$, from simulation and the values of bypassing vehicles required for the social optimality as a function of $f_i$.}\label{fig:social}
\end{figure*}

With the calibrated $\bb{C}$, we used our model to predict $x_1^s$, $x_1^b$, $x_2^s$, and $x_2^b$ for the scenarios where the total flow entering the diverge is different from the value that we used to calibrate our macroscopic diverge. Figure~\ref{fig:model_pred} demonstrates such a study when the total demand to the diverge was $d = 2500 \frac{\text{veh}}{\text{hour}}$. Figure~\ref{fig:model_pred} shows both the SUMO microsimulation results and our model predictions values as a function demand configuration fractions $f_i$'s. As Figure~\ref{fig:model_pred} shows, our macroscopic model accurately predicts the fraction of bypassing and steadfast vehicles for each destination. Notice that when the demand for exit $1$ is low $f_1 \leq 0.5$, none of the vehicles who aim to take the exit 1 would take the more crowded lane II; therefore, $x_1^b \simeq 0$. But, with the increase of $f_1$, vehicles will take lane II since it will reduce their cost. Our simulation results indicate that our model is capable of predicting , with great accuracy, the behavior of the vehicles. We obtained similar results when the total demand that enter the diverge was varied.

\section{Socially Optimal Lane Change Behavior}\label{sec:soc_opt}

Having shown that our model can be used to accurately predict macroscopic vehicular behavior, we can deploy it for further analysis. Intuitively, one might argue that if most vehicles were less selfish, and would have choose their destination lane far upstream of the diverge, the total travel cost experienced by vehicles at the diverge would be reduced. We now show that our model provides a powerful framework for analytically studying this conjecture. Assume that there is a central authority which can dictate the lanes that a vehicle must travel on such that the total vehicle cost at the diverge is minimum; or equivalently, that the social optimum is achieved. The total or social cost experienced by the vehicles can be computed using our desired macroscopic model as follows:

\begin{align}
J_{\text{soc}} = \sum_{i \in I} \left( x_i^s J_i^s + x_i^a J_i^b \right).
\end{align}

Then, the minimum social cost can be determined by solving the following optimization problem 

\begin{equation}\label{eq:social_opt}
\begin{aligned}
& \underset{\bb{x}}{\text{minimize}}
& & J_{\text{soc}} \\
& \text{subject to}
& & x_i^s + x_i^b = f_i, \quad \forall i \in I, \\
&&& x_i^s \geq 0, \; x_i^b \geq 0, \quad \forall i \in I.
\end{aligned}
\end{equation}
Optimization~\eqref{eq:social_opt} can be solved to find the optimal lane choice and bypassing behavior. Note that in~\eqref{eq:social_opt}, the decision variables are $x_i^s, x_i^b$, for every $i \in I$; thus, the objective function~\eqref{eq:social_opt} is a 3rd order polynomial in the decision variables. Optimization problem~\eqref{eq:social_opt} can be easily solved using commercial solvers. This simplicity should be contrasted to the existing methods, where strategies for finding better lane choices are heuristically determined through simulation studies.

Using the $\bb{C}$ that was obtained from our model calibration, we solved~\eqref{eq:social_opt} for the case when the total flow entering the diverge was $3000 \frac{\text{veh}}{\text{hour}}$.  Figure~\ref{fig:social} demonstrates the socially optimal bypassing of vehicles. As Figure~\ref{fig:social} shows, for every fractional demand configuration, at equilibrium, since vehicles choose their lanes selfishly, the number of bypassing vehicles is larger than the optimal one. Moreover, as Figure~\ref{fig:social} suggests, a key observation is that the optimal lane choice is \emph{not} preventing all vehicles from bypassing. Therefore, the socially optimal lane choice is  often in between the Wardrop equilibrium and zero lane change. Our model allows for quantitatively analyzing this trade off, which we believe has not been previously captured in the literature.



\section{Mixed Autonomy Setting}\label{sec:mixed_aut}

\begin{figure*}
    \centering
    \begin{subfigure}[b]{0.47\textwidth}
       \centering \includegraphics[width=\textwidth]{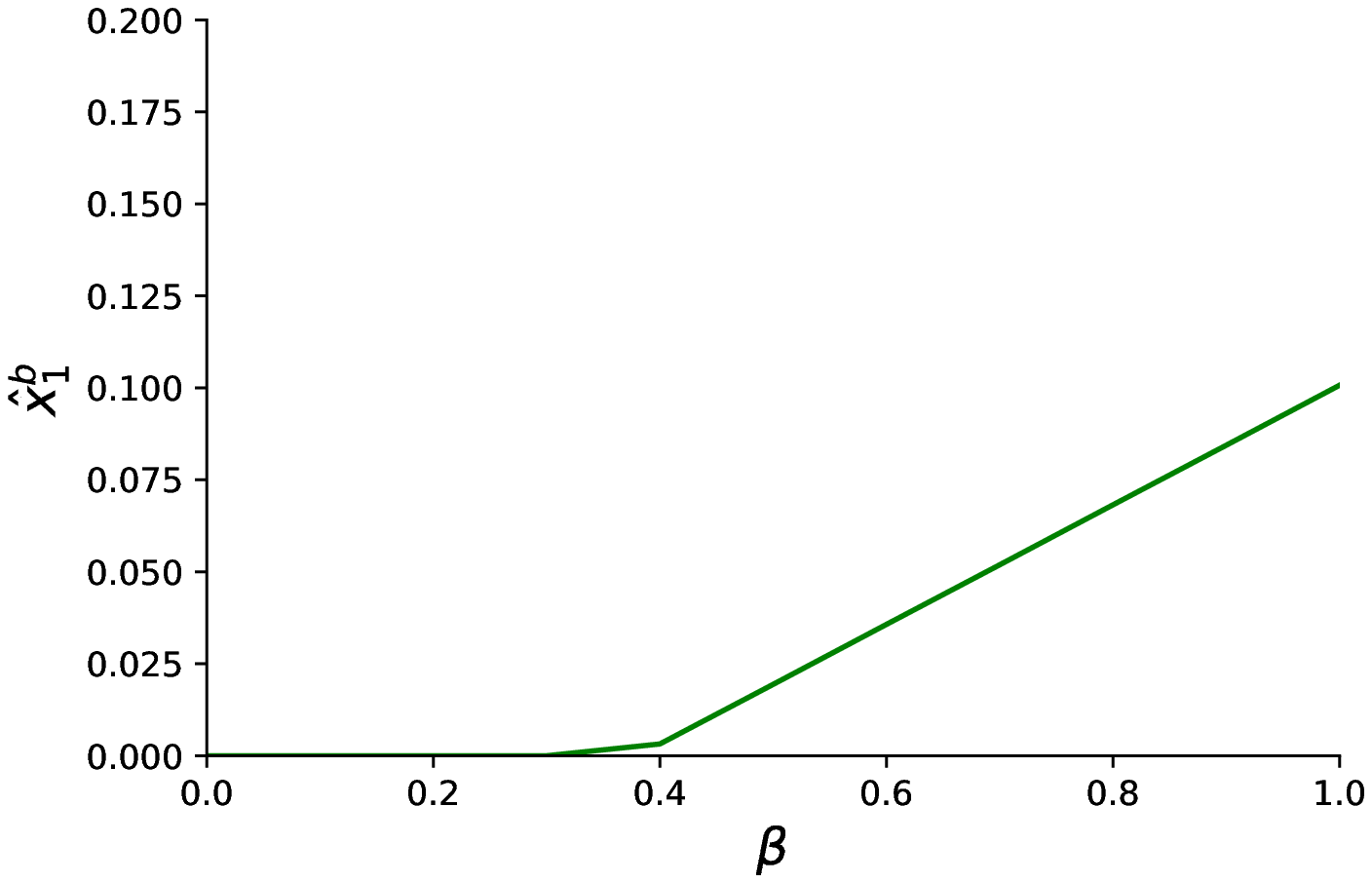}
        \caption{$\alpha = \% 25$}
    \end{subfigure}
    ~ 
      \begin{subfigure}[b]{0.47\textwidth}
       \centering \includegraphics[width=\textwidth]{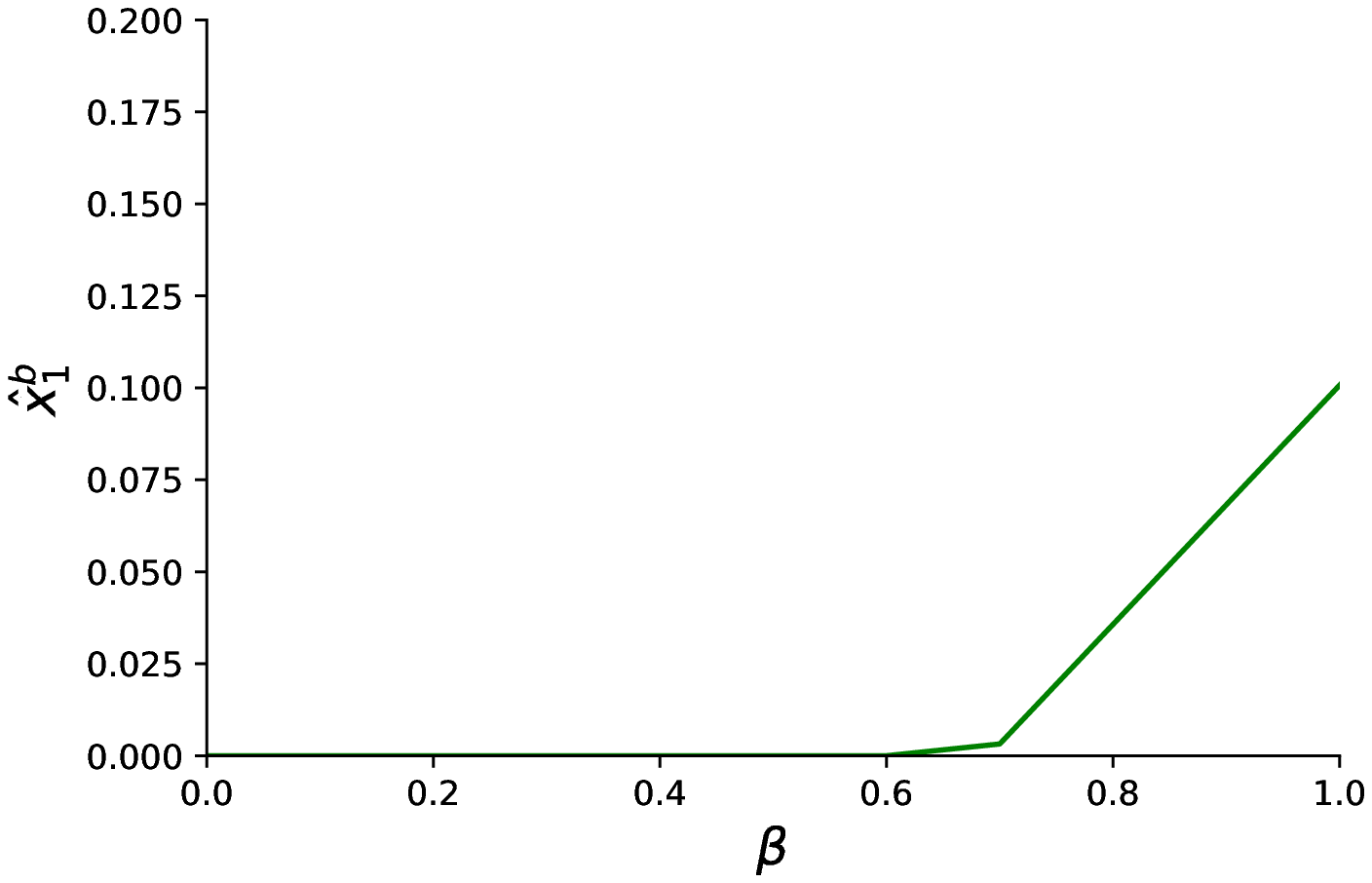}
        \caption{$\alpha = \% 50$}
    \end{subfigure}    
    
    \caption{The fraction of bypassing vehicles versus the fraction of vehicles that are commanded to remain steadfast $\beta$ for different values of autonomy fraction $\alpha$. }   \label{fig:xhat}
\end{figure*}

The significance of the mathematical macroscopic model that we derived in this paper is that it allows traffic engineers to perform further analysis, as well as to derive socially optimal traffic management policies. In particular,
a novel and important use of our model becomes apparent if we assume that a central authority has control over a fraction of vehicles, and is able to to dictate control actions and route choices to the compliant vehicles such that, when the noncompliant vehicles react to the actions of the compliant vehicles, the overall network performance is improved. In the context of routing games, such control mechanisms are known as Stackelberg routing of vehicles. Examples of such scenarios include traffic networks with mixed autonomy. As autonomous vehicles are becoming tangible technologies, it is expected that in near future both regular and autonomous vehicles will coexist in traffic networks. It was shown in~\cite{mehr2018can} that only replacing regular vehicles by autonomous vehicles might not be enough for improving network mobility, and further controlling mechanisms are required to exploit the mobility benefits of autonomous vehicles. Therefore, in order to increase the mobility of traffic networks with mixed autonomy, lane choice decisions for autonomous vehicles must be developed such that they lead to a decrease of the social cost. In this section, we use our model to study how the lane choice behavior of autonomous vehicles in the mixed autonomy setting can change the social cost at a traffic diverge with two exit links through an example.

Consider the traffic diverge shown in Figure~\ref{fig:diverge}. Assume the model has been calibrated with the cost parameter vector $\mathbf{C}$. Fix the total demand $d$ and the normalized demand configuration $F =\{ f_1 , f_2 \}$. Let $\alpha$ be the fraction of vehicles that are autonomous among all vehicles that wish to take exit I. We also assume in this example that all vehicles that take exit II are regular. This implies that $\alpha f_1$ percent of the total vehicles are autonomous. We also assume that the central authority is able to command $\beta$ fraction of autonomous vehicles to be steadfast vehicles, and the remaining $(1-\beta)$ fraction of autonomous vehicles to be bypassing vehicles i.e. $(1-\beta)$ fraction of autonomous vehicles are commanded to bypass at the diverge to take their exit link, and the remaining $\beta$ fraction of the autonomous vehicles are commanded to choose lane I far upstream of the diverge and remain on this lane. Let $w=(1-\beta) \alpha f_1$ and $z = \beta \alpha f_1$ denote the fractions with respect to the total demand of vehicles that are commanded to bypass and remain on lane I respectively. For fixed $w$ and $z$, the remaining vehicles react such that a new equilibrium is achieved. Thus, every choice of $w$ and $z$ induces a new Wardrop equilibrium. For each exit link $i \in I$, we use  $\hat{x}_i^s$ and $\hat{x}_i^b$ to represent the fraction of steadfast and bypassing vehicles in the induced equilibrium. Note that in this case, for a given $w$ and $z$, flow conservation requires that $\hat{x}_1^s + \hat{x}_1^b = f_1 - w - z$ and $\hat{x}_2^s + \hat{x}_2^b = f_2 $. Let $\hat{\mathbf{x}} = (\hat{x}_i^s, \hat{x}_i^b : i \in I)$ represent the vector of flows. In this case, the modified cost of steadfast and bypassing vehicles are

\begin{align*}
\hat{J}_1^s(\hat{\mathbf{x}}) &= C_1^t \left( \hat{x}_1^s + \hat{x}_2^b + z \right) + C_1^c (\hat{x}_1^b+w) \left ( \hat{x}_1^s + z + \hat{x}_2^b \right),\\
\hat{J}_1^b(\hat{\mathbf{x}}) &= C_2^t \left(\hat{x}_2^s + \gamma_1 (\hat{x}_1^b + w) \right) + C_2^c \hat{x}_2^b \left(  \hat{x}_2^s + \hat{x}_1^b + w \right), \\
\hat{J}_2^s(\hat{\mathbf{x}}) &= C_2^t \left( \hat{x}_2^s + \hat{x}_1^b + w \right) + C_2^c \hat{x}_2^b \left ( \hat{x}_2^s + \hat{x}_1^b + w \right),\\
J_2^b(\hat{\mathbf{x}}) &= C_1^t \left(\hat{x}_1^s + z + \gamma_2 \hat{x}_2 ^b \right) + C_1^c (\hat{x}_1^b + w) \left(  \hat{x}_1^s + z + \hat{x}_2^b \right). 
\end{align*}

\begin{figure*}
    \centering
    \begin{subfigure}[b]{0.47\textwidth}
       \centering \includegraphics[width=\textwidth]{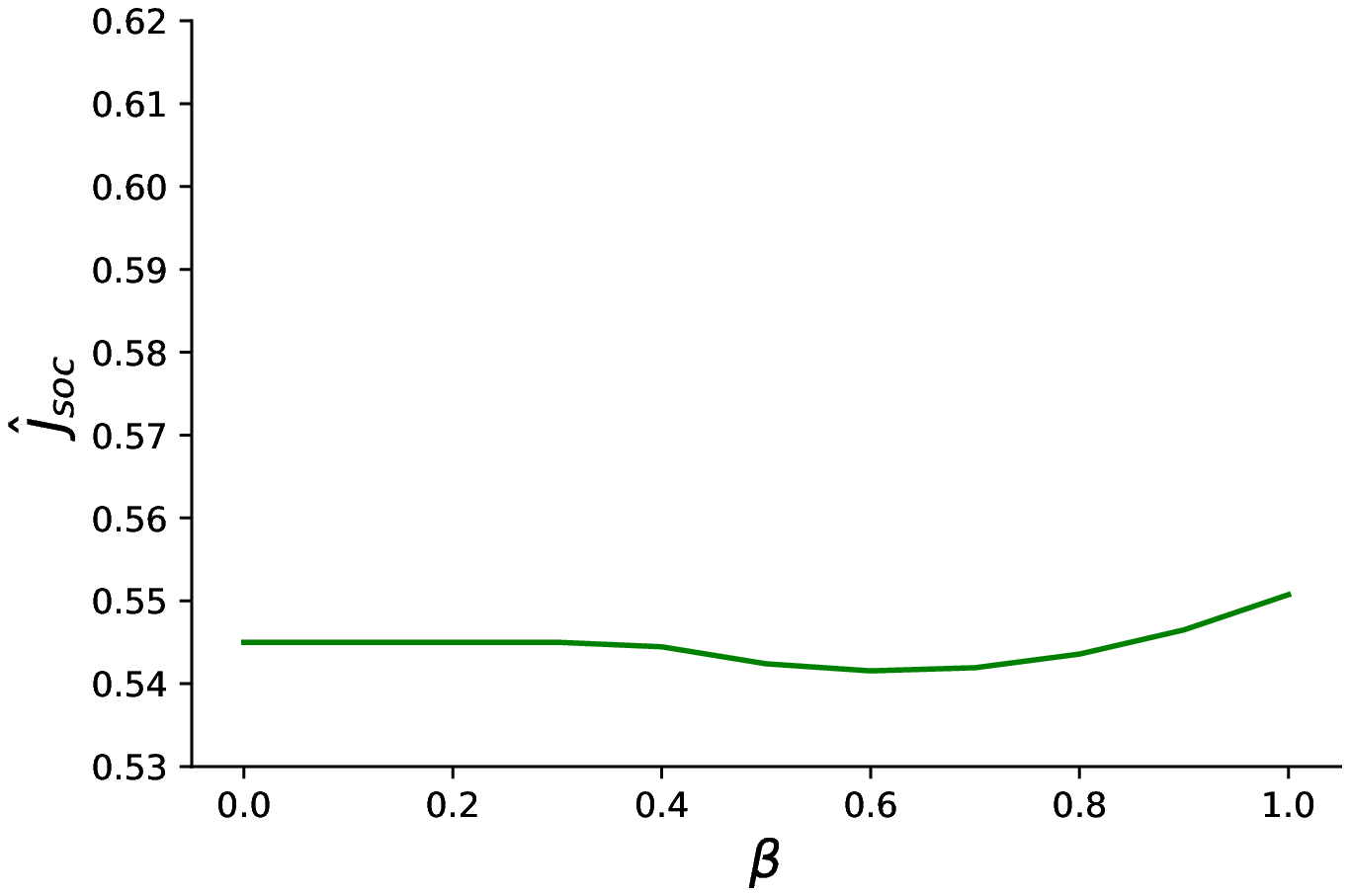}
        \caption{$\alpha = \% 25$}
    \end{subfigure}
    ~ 
      \begin{subfigure}[b]{0.47\textwidth}
       \centering \includegraphics[width=\textwidth]{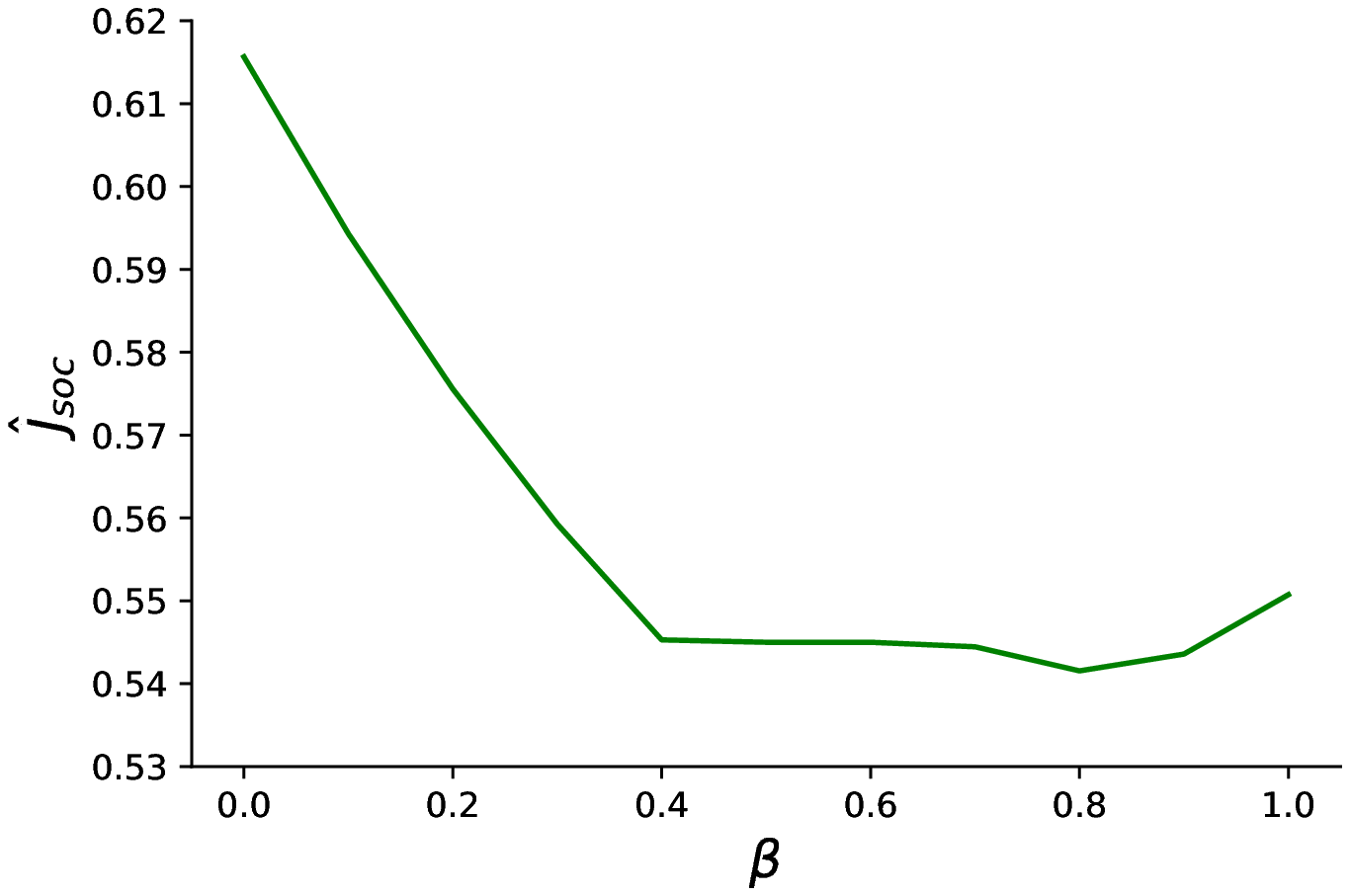}
        \caption{$\alpha = \% 50$}
    \end{subfigure}    
    
    \caption{The social cost per total demand versus the fraction of commanded bypassing vehicles $\beta$ for different autonomy fractions $\alpha$. }   \label{fig:social_red}
\end{figure*}

\noindent Using these modified cost functions, a flow vector $\hat{\mathbf{x}}$ is an induced equilibrium if for every exit link $i \in I$, it satisfies

\begin{subequations}
\label{eq:eq_ind_def}
\begin{gather}
\begin{align}
\hat{x}_i^s (\hat{J}_i^s(\hat{\mathbf{x}}) - \hat{J}_i^b(\hat{\mathbf{x}})) &\leq 0 ,\\
\hat{x}_i^b (\hat{J}_i^b(\hat{\mathbf{x}}) - \hat{J}_i^s(\hat{\mathbf{x}})) &\leq 0. 
\end{align}
\end{gather}
\end{subequations}

\noindent Let $\hat{J}_{\text{soc}}$ be the social cost of the users in the induced equilibrium. Then, the social cost $\hat{J}_{\text{soc}}(\hat{\mathbf{x}})$ is  obtained via

\begin{align}\label{eq:soc_cost_reduced}
\hat{J}_{\text{soc}}(\hat{\mathbf{x}}) = (\hat{x}_1^s + z) \hat{J}_1^s(\hat{\mathbf{x}}) + ( \hat{x}_1^b + w  ) \hat{J}_1^b(\hat{\mathbf{x}}) +
\hat{x}_2 \hat{J}_2^s(\hat{\mathbf{x}}) +  \hat{x}_2^b  \hat{J}_2^b(\hat{\mathbf{x}}).
\end{align}

To make our exposition more concrete, let the total demand be $d = 3000\, \frac{\text{veh}}{\text{hour}}$. We use the cost parameters $\mathbf{C}$ obtained from calibration. Now, fix the fractional demand configuration to be $F =\{ f_1  = 0.65, f_2 = 0.35\}$. We fix $\alpha$ and vary $\beta$ from 0 to 1. For each value of $\beta$, we computed the fractions of steadfast and bypassing vehicles in the induced equilibrium using Equations~\eqref{eq:eq_ind_def}. Figure~\ref{fig:xhat} demonstrates the predicted fractions of bypassing vehicles for different values of $\beta$. Then, using Equation~\eqref{eq:soc_cost_reduced}, we computed the resulting social cost as a function $\beta$. As it can be observed from Figure~\ref{fig:xhat}, in both cases, when $\beta = 0$, i.e. none of the  autonomous vehicles were steadfast, no bypassing was observed at the equilibrium. As $\beta$ increased to 0.4 for $\alpha = 0.25$ and $\beta$ increases to 0.7 for $\alpha = 0.5$, the vehicles started to bypass at the equilibrium. Intuitively, it means that when the fraction of commanded steadfast vehicles went beyond a threshold, the bypassing behavior of vehicles emerged. Figure~\ref{fig:social_red} plots the social cost versus $\beta$ for different values of $\alpha$. As this figure shows, when $\alpha = \% 0.25$, the minimum social cost was achieved around $\beta = \% 60$; however, when the fraction of autonomous vehicles was increased to $\alpha = \% 50$, the minimum social delay is obtained when $\beta  \in [0.8,0.9]$. Notice how the parameter $\alpha$ affects the pattern of the plots. For $\alpha = 0.25$, all values of $\beta$ lead to approximately similar values of the social cost, whereas for $\alpha = 0.5$, $\beta$ is more determinative. It is interesting to observe that the optimal social cost does not occur in scenarios when no vehicles bypass. The reason for this behavior is that when $\hat{x}^b_1 = 0$, the sacrifice in terms of the cost exerted on the commanded autonomous vehicles (which obey the commands of the central authority) was larger than the gains obtained in their lane choice. Social optimality is achieved when the improvement in the overall cost of regular vehicles minus the increases in the personal cost of commanded vehicles is maximized. Our model provides a powerful framework for performing this type of traffic analysis. In this section, we studied the impact of autonomy presence through an example. Our model can be further used to mathematically find the optimal lane choice of commanded vehicles for any given autonomous vehicles penetration rate.


\section{Conclusion} \label{sec:future}

We provided a game theoretic framework for macroscopically modeling the aggregate bypassing of vehicles at a traffic diverge, where vehicles were assumed to be selfish. We modeled the resulting equilibrium as a Wardrop equilibrium and proved the existence and uniqueness of this equilibrium. We described how our model can be easily calibrated and demonstrated via simulation studies that our model yielded promising results. We also showed how our model can be used to determine traffic management policies that reduce the social cost in traffic networks with mixed vehicle autonomy.
For future steps, we are hoping to verify our results using real traffic data.

\section*{Acknowledgment}

This work is supported by the National Science Foundation under Grants
CPS 1446145 and CPS 1545116.

\end{document}